\documentclass[conference]{IEEEtran}

\usepackage{amsmath}
\usepackage{amssymb}
\usepackage{amsthm}
\usepackage{graphicx}
\usepackage{float}
\usepackage{url}
\usepackage{tikz}
\usepackage{pgfplots}
\usepackage{grffile}
\usepackage{arydshln} 
\usepackage{comment}

\usetikzlibrary{positioning, calc,arrows,plotmarks}

\usepackage{lipsum} 

\newcommand\blfootnote[1]{%
  \begingroup
  \renewcommand\thefootnote{}\footnote{#1}%
  \addtocounter{footnote}{-1}%
  \endgroup
}

\newlength\figureheight
\newlength\figurewidth
\newtheorem{lemma}{Lemma}
\newtheorem{definition}{Definition}
\newtheorem{remark}{Remark}
\newtheorem{example}{Example}

\newtheorem{proposition}[lemma]{Proposition}
\newtheorem{fact}[lemma]{Fact}

\begin{document}

\title{Spatially Coupled Codes with Sub-Block Locality: Joint Finite Length-Asymptotic Design Approach\vspace{-0.3cm}}
\author{
    \IEEEauthorblockN{H. Esfahanizadeh$^{*\dagger}$, and E. Ram$^{*\ddagger}$, Y. Cassuto$^\ddagger$, and L. Dolecek$^\dagger$ }
     \IEEEauthorblockA{
        $^\ddagger$Andrew and Erna Viterbi Department of Electrical Engineering,
        Technion, Haifa, Israel\\
        \{s6eshedr@campus, ycassuto@ee\}.technion.ac.il
    }
    \IEEEauthorblockA{
        $^\dagger$University of California, Los Angeles (UCLA), Los Angeles, USA\\
      hesfahanizadeh@ucla.edu, dolecek@ee.ucla.edu \vspace{-0.6cm}
    }
}
\maketitle

\begin{abstract} 
SC-LDPC codes with sub-block locality can be decoded locally at the level of sub-blocks that are much smaller than the full code block, thus providing fast access to the coded information. The same code can also be decoded globally using the entire code block, for increased data reliability. In this paper, we pursue the analysis and design of such codes from both finite-length and asymptotic lenses. This mixed approach has rarely been applied in designing SC codes, but it is beneficial for optimizing code graphs for local and global performance simultaneously. Our proposed framework consists of two steps: 1) designing the local code for both threshold and cycle counts, and 2) designing the coupling of local codes for best cycle count in the global design.\vspace{-0.15cm}
\blfootnote{$*$ equal contribution}

\end{abstract}

\section{Introduction}
Spatially-coupled low-density parity-check (SC-LDPC) codes \cite{FelstromIT1999} are known to have many desired properties, such as threshold saturation \cite{KudekarRichUrb13} and linear-growth of minimal trapping sets of typical codes from the ensemble \cite{Mitchell11}. These properties imply good BER performance in the \emph{waterfall} and \emph{error floor} regions, using the Belief-Propagation (BP) decoder.
From these properties emerged a simple and effective design methodology: first choose a (usually regular) protograph with good MAP threshold \cite{MitchLentCost15}, and then optimize the coupling (edge spreading) for minimum incidence of short cycles \cite{EsfHar19}.

Recently, \cite{RamCass18b} introduced SC-LDPC codes with \emph{sub-block locality}, meaning that in addition to the usual full-block decoding, the codes can be decoded \emph{locally} in small sub-blocks for fast read access. Formally, a codeword from an ($L,n)$ code with \emph{sub-block locality}, i.e., $L,n\in\mathbb{N}$, consists of $L$ sub-blocks, each being a codeword of some \emph{local} code of length $n$ that can be decoded independently of other sub-blocks. The concatenation of the $L$ sub-blocks forms a codeword of a stronger \emph{global} code of length $Ln$ that provides higher data protection when needed \cite{CassHemo17, RamCass18A}. For these codes, it was shown in \cite{RamCass18b} that the existing design methodology of SC-LDPC codes is no longer sufficient, because one needs to optimize the code for the local decoding as well.

While \cite{RamCass18b} focused on asymptotic analysis of regular codes over the binary erasure channel, in this work, we optimize performance for both asymptotic and finite-length performances over the AWGN channel. Furthermore, we consider irregular local protographs to obtain superior performance \cite{Liva}. Our code design consists of two stages: local design, and global design (conditioned on the local code). In local design (Section~\ref{Sec:local_des}), we consider local irregular codes, comparing between two extreme options which we call \textit{balanced} and \textit{unbalanced}.
For the asymptotic analysis, we use the EXIT method \cite{TenBrink04}, and for the finite-length analysis, we build upon the combinatorial cycle-enumeration method from \cite{EsfHar19}. For some parameters our analytic results show a trade-off between asymptotic and finite-length performances, while for others the same choice optimizes both.
In global design (Section~\ref{Sec:CC_opt}), we extend the methods from \cite{EsfHar19} to optimize the global cycle incidence \emph{given the local code}. 
Simulation results (Section~\ref{Sec:Simulations}) show the conditional global design gives better performance than the existing coupling optimization without locality considerations.

\section{Preliminaries}
\subsection{SC-LDPC codes with sub-block locality}\label{Sub:SC-LDPCL}
An LDPC protograph is a small bipartite graph represented by a $\gamma\times\kappa$ bi-adjacency matrix $B$, i.e., there is an edge between check node (CN) $i$ and variable node (VN) $j$ if and only if $B_{i,j}=1$. In general, $B_{i,j}>1$ (parallel edges) is allowed, but in this work we focus on $B_{i,j}\in \{0,1\}$.
A sparse parity-check matrix $H$ (Tanner graph) is generated from $B$ by a lifting operation characterized by a positive integer $p$ that is called the circulant size. 
The rows (resp. columns) of $H$ corresponding to row $i\in\{0,1,\ldots,\gamma-1\}$ (resp. column $j\in\{0,1,\ldots,\kappa-1\}$) of $B$, are called row group $i$ (resp. column group $j$). For \emph{array-based} (AB) lifting \cite{FanTurbo2000}, the prime circulant size ensures no cycle-$4$ exists. Thus, this paper (partially) focuses on cycles-$6$.

Let $H$ be the parity-check matrix of an LDPC code. An SC-LDPC code \cite{FelstromIT1999} with memory $m$ and coupling-length $L$ is constructed from $H$ by partitioning it into $m+1$ matrices $H=\sum_{\tau=0}^{m}H_\tau$, and placing them $L$ times on the diagonal of the coupled parity-check matrix $H_{SC}$.
In this work, we focus on $m=1$ SC codes, thus the partitioning operation determines which (non-zero) circulant is assigned to $H_0$ and which one is assigned to $H_1$ (when referring to protographs, we use $B_0$ and $B_1$). We represent this partitioning by a \textit{ternary} matrix $P$, where $P_{i,j}\in\{0,1,X\}$. If $P_{i,j}=X$, then there is a $p\times p$ zero matrix in row group $i$ and column group $j$ of $H$. Otherwise, the non-zero circulant is assigned to $H_{P_{i,j}}$. This description captures SC constructions from both regular and irregular protographs.

For local decoding, only CNs that are not connected to VNs outside this sub-block can help. We call these CNs \emph{local CNs} (LCNs). All other CNs are called \emph{coupling CNs} (CCNs) \cite{RamCass18b}. In terms of partitioning $H$, rows in $P$ that have both elements $1$ and $0$ result in CCNs in the coupled matrix; we mark the number of such rows as $\gamma_C$. An SC-LDPC code with sub-block locality is constructed by constraining the partitioning such that $\gamma_L\triangleq \gamma-\gamma_C$ rows in $P$, corresponding to LCNs, lead to a non-zero asymptotic (local) decoding threshold  \cite{RamCass18b}. Without loss of generality, the rows of $P$ are ordered such that the first $\gamma_C$ rows correspond to CCNs. 

\subsection{Asymptotic Analysis of Protographs: The EXIT method}
The EXtrinsic Information Transfer (EXIT) method \cite{TenBrink04} is a useful tool for analyzing and designing LDPC codes in the asymptotic regime over the AWGN channel with channel parameter $\sigma$. Let $ J\colon [0,\infty)\to [0,1) $ be a function that represents the mutual information between the channel input and a corresponding message passing in the Tanner graph, and let $s\in[0,\infty)$ be the message's standard deviation.
For a VN of degree $d_v$ in the protograph, with incoming EXIT values $\{J_{i}\}_{i=1}^{d_v-1}$, the VN$\rightarrow$CN EXIT value is given by
\begin{align}\label{Eq:J VN}
J_{out}^{(V)}\!\left (s_{ch},J_{1},\ldots,J_{d_v\!-\!1}\right )\!=\!J\!\!\left (\!\sqrt{\sum_{i=1}^{d_v-1}\!\!\left (J^{-1}(J_{i})\right )^2\!+\!s_{ch}^2}\right )\!,
\end{align}
where $s_{ch}^2=4/\sigma^2$.
For a CN of degree $d_c$ in the protograph with incoming EXIT values $\{J_{j}\}_{j=1}^{d_c-1}$, the CN$\rightarrow$VN EXIT value is
\begin{align}\label{Eq:J CN}
\begin{split}
J_{out}^{(C)}\!(J_{1},\ldots,J_{d_c-1})\! =\! 1\!-\!J_{out}^{(V)}\!\left (0,1\!-\!J_{1},\ldots,1\!-\!J_{d_c\!-\!1}\right ).
\end{split}
\end{align}
\looseness=-1
In simulations, we use approximations of $J(\cdot)$ and $J^{-1}(\cdot)$  \cite{TenBrink04}. The functions $J_{out}^{(V)}$ and $J_{out}^{(C)}$ are monotonically non decreasing with respect to all their arguments.
By alternately applying \eqref{Eq:J VN} and \eqref{Eq:J CN} for every edge in a protograph with varying values of the channel parameter $ \sigma$, a threshold value $ \sigma^*$ can be found, such that all EXIT values on VNs approach $ 1 $ as number of iterations increases if and only if $ \sigma<\sigma^* $ \cite{Liva}. We mark the threshold of a protograph $B $ by $ \sigma^*(B) $.
Given two protographs, it is not clear, in general, which one yields a better threshold since many parameters need to be tracked. However, in some cases we can order the thresholds of two protographs. The following ordering of regular protographs (with different rates) will be used in the sequel as a supporting lemma.
\begin{fact}\label{Fact: Threshold monotonicity}
Let $\sigma^*_1$ and $\sigma^*_2$ be the EXIT thresholds of  $(\gamma_1,\kappa_1)$-regular and $(\gamma_2,\kappa_2)$-regular protographs.  If $\kappa_1=\kappa_2$ and $\gamma_1\leq\gamma_2$, then $\sigma^*_1\leq\sigma^*_2$, and if $\gamma_1=\gamma_2$ and $\kappa_1\leq\kappa_2$, then $\sigma^*_1\geq\sigma^*_2$.
\end{fact}

\subsection{Short-Cycle Optimization}
\label{Sub:cycles_opt}
Short cycles have a negative impact on the performance of block-LDPC and SC-LDPC codes under BP decoding: 1) they affect the independence of messages that are transferred on the graph, 2) they enforce upper-bounds on the minimum distance, and 3) they form combinatorial objects in the Tanner graphs that are known to be problematic \cite{SmarandacheIT2012,EsfHar19}.
\begin{definition}\label{def_ov_par}
Consider a binary matrix $R$. A degree-$d$ overlap parameter $t_{\{i_1,\dots,i_d\}}$ is the number of columns in which all rows of $R$ indexed by ${\{i_1,\dots,i_d\}}$ have $1$s.
\end{definition}

The overlap parameters of the matrix $R=[B_0^\textnormal{T}\hspace{0.1cm}B_1^\textnormal{T}]^\textnormal{T}$ (known as replica \cite{EsfHar19}) of size $2\gamma\times\kappa$ contains all information we need to find the number of cycles in the corresponding SC code's protograph. We are particularly interested in cycles-$6$, as they are the shortest cycles for practical LDPC codes (most practical high-rate LDPC codes, in particular the codes in this paper, are designed with girth 6). The set of non-zero overlap parameters is:\vspace{-0.1cm}
\begin{equation}\label{equ_nz_ov_par}
\begin{split}
\mathcal{O}=\{t_{\{i_1,\dots,i_d\}}\textnormal{ }|\textnormal{ }1\leq d\leq \gamma,0\leq i_1,\dots,i_d<2\gamma,\\ \forall \{i_u,i_v\}\subset\{i_1,\dots,i_d\}\hspace{0.1cm}i_u\neq i_v \textnormal{ (mod $\gamma$)}\}.
\end{split}
\end{equation}
The overlap parameters in (\ref{equ_nz_ov_par}) are not all independent. The set of all \emph{independent} non-zero overlap parameters is
$\mathcal{O}_\textnormal{ind}=\{t_{\{i_1,\dots,i_d\}}\textnormal{ }|\textnormal{ }1\leq d\leq \gamma, 0\leq i_1,\dots,i_d<\gamma\}$, \cite{EsfHar19}.
The number of cycles-$6$ in the protograph of an SC code with parameters $m=1$, $L$, and $\mathcal{O}_{\textnormal{ind}}$ is given by $F=LF_1(\mathcal{O}_{\textnormal{ind}})+(L-1)F_2(\mathcal{O}_{\textnormal{ind}})$, where $F_1$ and $F_2$ are the number of cycles-$6$ that span one and two replica(s) of the coupled protograph, respectively, and they are determined solely as functions of overlap parameters.

The discrete optimization problem of minimizing cycles-$6$ is $F^*=\min_{\mathcal{O}_\textnormal{ind}}F$. 
This optimization for identifying the optimal overlap set, and consequently the optimal partitioning, results in the minimum number of cycles-$6$ in an SC protograph \cite{EsfHar19}. The approach is called the optimal overlap (OO) partitioning. In this paper, we customize the OO partitioning in order to design SC-LDPC codes with sub-block locality for local and global decoding.
One contribution of our paper (Section \ref{Sec:CC_opt}) is to extend OO partitioning to find the optimal design of CCNs conditioned on existence of a specific number of LCNs.

\section{Local Design}
\label{Sec:local_des}
In this section, we propose two protograph constructions for local codes of a SC-LDPC code with sub-block locality and parameters $\gamma_L$, $\kappa$, and $\nu$, where $\nu\in[0,\kappa-1]$ is the number of zero circulants per local code. The two local code designs we propose both have the same rate but stand at two ends of the spectrum of irregular designs with the given parameters. We first define some matrices that are used in the constructions.

For integers $l$, $k$, and $i$ such that $0\leq i < l$, let $Q(l,k;i)$ and $S(l,k)$ be $l\times k$  matrices, such that 
\begin{align*}
 &\left[Q(l,k;i)\right]_{s,t} = \left\{ 
    \begin{array}{ll}
         0&s=i  \\
         1&\text{otherwise},
    \end{array}\right.\\
&\left[S(l,k)\right]_{s,t} = \left\{ 
    \begin{array}{ll}
        0&s\in[0,k)\,,t=k-s-1  \\
        1& \text{otherwise}.
    \end{array}\right.
\end{align*}
Let $\mathbf{1}(l,k)$ be an all-one matrix and $\mathbf{0}(l,k)$ be an all-zero matrix with size $l\times k$, and let $\nu=a\gamma_L+b$ with integers $a,b$ such that $0\leq b < \gamma_L$. The \textit{balanced} and \textit{unbalanced} local code constructions are represented by the protograph matrices $B_\mathcal{B}$ and $B_\mathcal{U}$, respectively, and defined as follows:\vspace{-0.2cm}
\begin{flalign}\label{def_mat_BB}
B_\mathcal{B}\hspace{-0.08cm}=
{\small\left(\!\!
\begin{array}{c:c:c:c:c}
\!\mathbf{1}(\gamma_L,\kappa-\nu)\hspace{-0.08cm}&\hspace{-0.08cm} S(\gamma_L,b)\hspace{-0.08cm}&\hspace{-0.08cm} Q(\gamma_L,a;\gamma_L\!-\!1)\hspace{-0.08cm}&\hspace{-0.08cm}\dots\hspace{-0.08cm}&\hspace{-0.08cm} Q(\gamma_L,a;0)\!
\end{array}\!\! \right)}
\end{flalign}\vspace{-0.5cm}
\begin{flalign}\label{def_mat_BU}
B_\mathcal{U}={\small\left(\!\!
\begin{array}{c:c}
\mathbf{1}(\gamma_L,\kappa-\nu) & Q(\gamma_L,\nu;0) 
\end{array}\!\! \right),}
\end{flalign}
where the vertical dashed lines represent the horizontal concatenation of sub-matrices. $B_\mathcal{B}$ and $B_\mathcal{U}$ are both $\gamma_L \times \kappa$ matrices with $\nu$ zero entries; in $B_\mathcal{B}$, zeros are uniformly distributed among the rows, while in $B_\mathcal{U}$, all zeros are in the first row.

\begin{example}\label{Ex:Const 1}
	Let $ \gamma_L=3 $, $\kappa=13$, and $ \nu = 10 $. Then,
	\begin{align*}
	&B_\mathcal{B} = 
	\left(\begin{array}{c:c:c:c:c}
	1\,1\,1&0&1\,1\,1&1\,1\,1&0\,0\,0\\
	1\,1\,1&1&1\,1\,1&0\,0\,0&1\,1\,1\\
	1\,1\,1&1&0\,0\,0&1\,1\,1&1\,1\,1
	\end{array}\right),\\
	&B_\mathcal{U} = 
	\left(\begin{array}{c:cccc}
    1\,1\,1&0&0\,0\,0&0\,0\,0&0\,0\,0 \\
    1\,1\,1&1&1\,1\,1&1\,1\,1&1\,1\,1\\
    1\,1\,1&1&1\,1\,1&1\,1\,1&1\,1\,1
    \end{array} \right).
	\end{align*}
\end{example}

\subsection{Threshold Derivations}

\begin{proposition}\label{Prop:thU<thB}
	Let $\kappa$, $\gamma_L$, and $\nu<\kappa$ be positive integers. If $ \kappa-\lfloor\tfrac{\nu}{\gamma_L}\rfloor \leq \nu$, then $ \sigma^*(B_\mathcal{U})\leq  \sigma^*(B_\mathcal{B}) $.
\end{proposition}

\begin{proof}
Let $\nu=a\gamma_L +b$. Consider a $ (\gamma_L-1,\kappa-a) $-regular protograph. Assume that we apply \eqref{Eq:J VN} and \eqref{Eq:J CN} on this protograph, and let $ x_\ell(\sigma) $ and $ u_\ell(\sigma) $ denote the resulting VN$\rightarrow$CN and CN$\rightarrow$VN EXIT values at iteration $ \ell $, respectively, given the channel parameter $\sigma$. We construct a $\gamma_L\times(2\kappa-\nu+b)$ protograph matrix $\hat{B}_\mathcal{B}$ as follows:
\begin{align*} 
\hat{B}_\mathcal{B}=&\left(\begin{array}{l:l:}
\mathbf{1}(\gamma_L,\kappa\!-\!\nu)\!-\! Q(\gamma_L,\kappa\!-\!\nu,0)& 
\mathbf{1}(\gamma_L,b)\!-\!S(\gamma_L,b)
\end{array}\right. \\
&\hspace{3mm}\begin{array}{l:l:}
Q(\gamma_L,\kappa\!-\!\nu,0)&
S(\gamma_L,b)\end{array} \\
&\hspace{2mm}\left.\begin{array}{l:l:l}
 Q(\gamma_L,a,\gamma_L-1) &
\dots &
Q(\gamma_L,a,0) 
\end{array} \right).
\end{align*}
In other words, $\hat B_\mathcal{B}$ is obtained from $B_\mathcal{B}$ by: 1) replacing the leftmost $\kappa-\nu$ entries in the first row with zeros such that all VNs are $\gamma_L-1$ regular, and 2) adding $\kappa-\nu+b$ columns of degree $1$ such that all CNs are $\kappa-a$ regular. We call the added degree-$1$ columns (VNs) ``auxiliary VNs" (see Fig.~\ref{Fig:Th_Proposition} for an example with $\kappa=5,\gamma_L=3,\nu=4$). Thus,  $\hat{B}_\mathcal{B}$ is $(\gamma_L-1,\kappa-a)$-regular except for the auxiliary VNs.
Next, we apply \eqref{Eq:J VN} and \eqref{Eq:J CN} on $\hat{B}_\mathcal{B}$ with a channel parameter $\sigma$ for non-auxiliary VNs, while the auxiliary VNs pass through a channel with a parameter $\sigma_\ell $ that changes in every iteration $\ell$ in a way that $ J(\sigma_\ell)=x_\ell(\sigma)$. It follows that the EXIT values passing over all edges of $ \hat B_\mathcal{B} $ equal to those passing over a $ (\gamma_L-1,\kappa-a) $-regular protograph, i.e., $ x_\ell(\sigma) $ and $ u_\ell(\sigma) $ for VN$\rightarrow$CN and CN$\rightarrow$VN messages, respectively. We match the edges in $B_\mathcal{B}$ to the edges in $\hat{B}_\mathcal{B}$ as follows. The edges connecting the $\nu$ rightmost columns in $B_\mathcal{B}$ match their identical edges in $\hat{B}_\mathcal{B}$, and the edges connecting bottom-most $\gamma_L-1$ CNs with the leftmost $\kappa-\nu$ VNs in $B_\mathcal{B}$ match their identical edges in $\hat{B}_\mathcal{B}$ as well. Finally, the edges connecting the top CN with the leftmost $\kappa-\nu$ VNs each matches one arbitrary edge connected to an auxiliary VN (see Fig.~\ref{Fig:Th_Proposition}).
\begin{figure}
    \centering
    \begin{tikzpicture}[>=latex]
\tikzstyle{cnode}=[rectangle,draw,fill=gray!70!white,minimum size=4mm]
\tikzstyle{vnode}=[circle,draw,fill=gray!70!white,minimum size=3mm]
\pgfmathsetmacro{\x}{1.7}
\pgfmathsetmacro{\y}{1.5}

	\foreach \c in {0,1,2}
	{
		\node[cnode] (c\c) at (\c*\x,0) {\tiny$\c$};	
	}
    \foreach \v in {1,2,3,4}
	{
		\pgfmathtruncatemacro{\k}{\v -1}
		\node[vnode] (v\v) at (\k*\x*2/3,-\y) {\tiny$\v$};	
		\node (ch\v) [below=2*\y mm of v\v] {\footnotesize $\sigma$};
		\draw [thick,->] (ch\v)--(v\v);
	}
	\node[vnode] (v0) at (\x,0.75*\y) {\tiny $0$};	
    \node (ch0) [above=2*\y mm of v0] {\footnotesize $\sigma$};
    \draw [thick,->] (ch0)--(v0);
    \draw [thick] (c0)--(v2) node[fill=white,inner sep=1pt,pos = 0.24] {\footnotesize $e_4$};
    \draw [thick] (c0)--(v3) node[fill=white,inner sep=1pt,pos = 0.24] {\footnotesize $e_5$};
    \draw [thick] (c0)--(v0) node[fill=white,inner sep=1pt,pos = 0.7] {\footnotesize $e_1$};
    \draw [thick] (c1)--(v0) node[fill=white,inner sep=1pt,pos = 0.7] {\footnotesize $e_2$};
    \draw [thick] (c1)--(v1) node[fill=white,inner sep=1pt,pos = 0.2] {\footnotesize $e_6$};
    \draw [thick] (c1)--(v2) node[fill=white,inner sep=1pt,pos = 0.2] {\footnotesize $e_7$};
    \draw [thick] (c1)--(v4) node[fill=white,inner sep=1pt,pos = 0.16] {\footnotesize $e_8$};
    \draw [thick] (c2)--(v0) node[fill=white,inner sep=1pt,pos = 0.7] {\footnotesize $e_3$};
    \draw [thick] (c2)--(v1) node[fill=white,inner sep=1pt,pos = 0.2] {\footnotesize $e_9$};
    \draw [thick] (c2)--(v3) node[fill=white,inner sep=1pt,pos = 0.2] {\footnotesize $e_{10}$};
    \draw [thick] (c2)--(v4) node[fill=white,inner sep=1pt,pos = 0.4] {\footnotesize $e_{11}$};
    \node [below=14*\y mm of c1] {\footnotesize(a)};
    \begin{scope}[shift={(2.8*\x,0)}]
    \foreach \c in {0,1,2}
	{
		\node[cnode] (c\c) at (\c*\x,0) {\tiny $\c$};	
	}
    \foreach \v in {1,2,3,4}
	{
		\pgfmathtruncatemacro{\k}{\v -1}
		\node[vnode] (v\v) at (\k*\x*2/3,-\y) {\tiny $\v$};	
		\node (ch\v) [below=2*\y mm of v\v] {\footnotesize $\sigma$};
		\draw [thick,->] (ch\v)--(v\v);
	}
	\node[vnode] (v0) at (\x,0.75*\y) {\tiny $0$};	
    \node (ch0) [above=2*\y mm of v0] {\footnotesize $\sigma$};
    \draw [thick,->] (ch0)--(v0);

    \node (z) [inner sep = 0pt]at (c0|-v0) {};
    \node[vnode,left = 1mm of z,inner sep = 1pt] (a0) {\tiny $A_0$};	
    \node (cha0) [above=2*\y mm of a0] {\footnotesize $\sigma_\ell$};
    \draw [thick,->] (cha0)--(a0);
    \node[vnode,right = 1mm of z,inner sep = 1pt] (a1) {\tiny $A_1$};	
    \node (cha1) [above=2*\y mm of a1] {\footnotesize $\sigma_\ell$};
    \draw [thick,->] (cha1)--(a1);
    \draw [thick] (c0)--(a1) node[fill=white,inner sep=1pt,pos = 0.7] {\footnotesize $e_1$};
    \draw [thick] (c0)--(a0);
    \draw [thick] (c0)--(v2) node[fill=white,inner sep=1pt,pos = 0.24] {\footnotesize $e_4$};
    \draw [thick] (c0)--(v3) node[fill=white,inner sep=1pt,pos = 0.24] {\footnotesize $e_5$};
    \draw [thick] (c1)--(v0) node[fill=white,inner sep=1pt,pos = 0.7] {\footnotesize $e_2$};
    \draw [thick] (c1)--(v1) node[fill=white,inner sep=1pt,pos = 0.2] {\footnotesize $e_6$};
    \draw [thick] (c1)--(v2) node[fill=white,inner sep=1pt,pos = 0.2] {\footnotesize $e_7$};
    \draw [thick] (c1)--(v4) node[fill=white,inner sep=1pt,pos = 0.16] {\footnotesize $e_8$};
    \draw [thick] (c2)--(v0) node[fill=white,inner sep=1pt,pos = 0.7] {\footnotesize $e_3$};
    \draw [thick] (c2)--(v1) node[fill=white,inner sep=1pt,pos = 0.2] {\footnotesize $e_9$};
    \draw [thick] (c2)--(v3) node[fill=white,inner sep=1pt,pos = 0.2] {\footnotesize $e_{10}$};
    \draw [thick] (c2)--(v4) node[fill=white,inner sep=1pt,pos = 0.4] {\footnotesize $e_{11}$};
    \node (b) [below=14*\y mm of c1] {\footnotesize(b)};
    \end{scope}
    \node (Y1) [left=1mm of cha0] {};
    \node (Y0) at(Y1|-b) {};
    \draw [dotted] (Y1)--(Y0);

    \end{tikzpicture}
    \caption{Graph constructions for Proposition~\ref{Prop:thU<thB}'s proof, with $\kappa=5,\gamma_L=3,\nu=4$: (a) corresponds to $B_\mathcal{B}$ and (b) corresponds to $\hat{B}_\mathcal{B}$. $A_0$ and $A_1$ are auxiliary VNs. The edge matching is illustrated via edge labels $\{e_i\}_{i=1}^{11}$.}
    \label{Fig:Th_Proposition}
\end{figure}
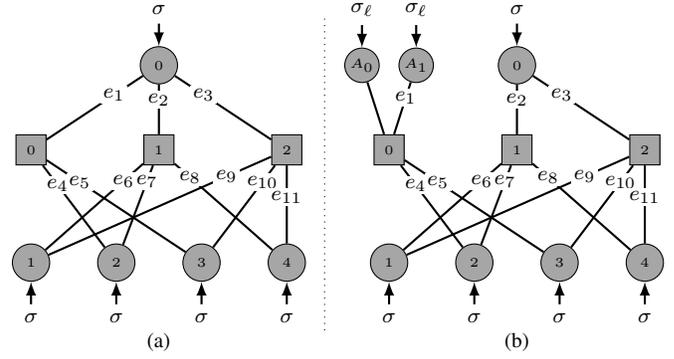
Given a channel parameter $ \sigma $, let $ y_\ell(\sigma,e) $ and $ w_\ell(\sigma,e) $ be the VN$\rightarrow$CN and CN$\rightarrow$VN EXIT values, respectively, over some edge $ e $ in the protograph $ B_\mathcal{B} $. 
From the monotonicity of \eqref{Eq:J VN} and \eqref{Eq:J CN} in their arguments and in node degrees, it can be shown by mathematical induction that for any $\sigma$ and every edge $ e $ 
\begin{align}
\label{Eq:y>x}
y_\ell(\sigma,e) \geq x_\ell(\sigma),\quad w_\ell(\sigma,e) \geq u_\ell(\sigma),\quad \forall \ell\geq 0\,.
\end{align}

If we mark $ \sigma^*(d_v,d_c) $ as the asymptotic threshold of a regular $ (d_v,d_c) $ protograph, then \eqref{Eq:y>x} implies that if the channel parameter satisfies $ \sigma<\sigma^*(\gamma_L-1,\kappa-a) $ then the EXIT algorithm over $ B_\mathcal{B} $ will converge to $ 1 $, thus\vspace{-0.12cm} 
\begin{align}\label{Eq:H_B>}
 \sigma^*(B_\mathcal{B}) \geq \sigma^*(\gamma_L-1,\kappa-a)\,.
\end{align}
From the sub-matrix lemma in \cite[Lemma~1]{RamCass18b} we have\vspace{-0.12cm} \begin{align}\label{Eq:H_U<}
\sigma^*(B_\mathcal{U})\leq \sigma^*(\gamma_L-1,\nu)\,.
\end{align}
Since $ \kappa-a\leq \nu $, combining \eqref{Eq:H_B>}--\eqref{Eq:H_U<} with Fact~\ref{Fact: Threshold monotonicity}, which holds since $\kappa-a = \kappa-\lfloor\nu/\gamma_L\rfloor\leq \nu$, completes the proof. 
\end{proof}

\subsection{Cycle Properties}

\begin{proposition}\label{Prop:C6U<C6B}
	 Let $\gamma_L\!=\!3$, $\kappa\!>\!0$, and $\nu\!=\!a\gamma_L\!+\!b\!<\!\kappa$ (where $a\geq 0$ and $0\leq b<\gamma_L$), and let $F(B_\mathcal{B})$ and $F(B_\mathcal{U})$ denote the number of cycles-$6$ in the protograph of the balanced and unbalanced local codes, respectively. Then $F(B_\mathcal{U})\leq  F(B_\mathcal{B})$. 
\end{proposition}
\begin{proof}
Consider any matrix $B$ of a local protograph with $\gamma_L=3$. The number of cycles-$6$ can be expressed in terms of the overlap parameters of matrix $B$ as follows: 
\begin{equation*}
    F(B)=\mathcal{A}(t_{\{0,1,2\}},t_{\{0,1\}},t_{\{0,2\}},t_{\{1,2\}})\,,
\end{equation*}
where $\mathcal{A}$ is given by (see \cite{EsfHar19})
\begin{equation}\label{equ_A}\vspace{-0.0cm}
\begin{split}
\mathcal{A}&(t_{\{i_1,i_2,i_3\}},t_{\{i_1,i_2\}},t_{\{i_1,i_3\}},t_{\{i_2,i_3\}})\\
=&\left(t_{\{i_1,i_2,i_3\}}[t_{\{i_1,i_2,i_3\}}-1]^+[t_{\{i_2,i_3\}}-2]^+\right)\\
+&\left(t_{\{i_1,i_2,i_3\}}(t_{\{i_1,i_3\}}-t_{\{i_1,i_2,i_3\}})[t_{\{i_2,i_3\}}-1]^+\right)\\
+&\left((t_{\{i_1,i_2\}}-t_{\{i_1,i_2,i_3\}})t_{\{i_1,i_2,i_3\}}[t_{\{i_2,i_3\}}-1]^+\right)\\
+&\left((t_{\{i_1,i_2\}}-t_{\{i_1,i_2,i_3\}})(t_{\{i_1,i_3\}}-t_{\{i_1,i_2,i_3\}})t_{\{i_2,i_3\}}\right)\,,\vspace{-0.2cm}
\end{split}
\end{equation}
and, $[\alpha]^+=\max\{\alpha,0\}$.
According to our constructions, no two zeros (out of the $\nu$ zeros) are located in the same column, thus
$ t_{\{0,1,2\}} = \kappa-\nu=\kappa-3a-b\geq 1$.
In the balanced construction, we have $t_{\{0,1\}}=\kappa-2a-b$, $t_{\{0,2\}}=\kappa-2a-(b>0)$, and $t_{\{1,2\}}=\kappa-2a-(b>1)$, where $(\mathrm{cond})$ is $1$ if $\mathrm{cond}$ is true and $0$ otherwise. Thus,\vspace{-0.1cm}
\begin{equation}\label{Eq:F_Balgamma3}
\begin{split}
    F(B_\mathcal{B})=&(\kappa-\nu)(\kappa-\nu-1)(\kappa-2a-(b>1)-2)\\
    +&(\kappa-\nu)(a+(b>1))(\kappa-2a-(b>1)-1)\\
    +&a(\kappa-\nu)(\kappa-2a-(b>1)-1)\\
    +&a(\kappa-\nu)(\kappa-2a-(b>1)).\vspace{-0.2cm}
\end{split}\vspace{-0.2cm}
\end{equation}
In the unbalanced construction, we have $t_{\{0,1\}}=t_{\{0,2\}}=\kappa-\nu$ and  $t_{\{1,2\}}=\kappa$. Thus,\vspace{-0.1cm}
\begin{align}\label{Eq:F_Unbalgamma3}
    F(B_\mathcal{U}) =(\kappa-\nu)(\kappa-\nu-1)(\kappa-2)\,.\vspace{-0.2cm}
\end{align}
Comparing \eqref{Eq:F_Balgamma3} and \eqref{Eq:F_Unbalgamma3} completes the proof.\vspace{-0.2cm}
\end{proof}

\begin{proposition}\label{pro_FB_FU_gamma4}
	 For $\gamma_L=4$, $\kappa>0$, and $\nu=a\gamma_L<\kappa$ (where $a>0$), the cycle-6 counts satisfy $F(B_\mathcal{U})>F(B_\mathcal{B})$. 
\end{proposition}
\begin{proof}
Consider a local  protograph $B$ with $\gamma_L=4$. The number of cycles-$6$ in $B$ is given by
\begin{equation*}
\begin{split}
F(B)
&=\mathcal{A}(t_{\{0,1,2\}},t_{\{0,1\}},t_{\{0,2\}},t_{\{1,2\}})\\
&+\mathcal{A}(t_{\{0,1,3\}},t_{\{0,1\}},t_{\{0,3\}},t_{\{1,3\}})\\
&+\mathcal{A}(t_{\{0,2,3\}},t_{\{0,2\}},t_{\{0,3\}},t_{\{2,3\}})\\
&+\mathcal{A}(t_{\{1,2,3\}},t_{\{1,2\}},t_{\{1,3\}},t_{\{2,3\}}).
\end{split}
\end{equation*}

Again, zeros are never located in the same column according to our constructions. In the balanced construction, we have $t_{\{0,1\}}=t_{\{0,2\}}=t_{\{0,3\}}=t_{\{1,2\}}=t_{\{1,3\}}=t_{\{2,3\}}=\kappa-\nu/2 $, and $t_{\{0,1,2\}}=t_{\{0,1,3\}}=t_{\{0,2,3\}}=t_{\{1,2,3\}}=\kappa-3\nu/4$, thus
\begin{align}
\label{Eq:F_Balgamma4}
\begin{split}
    F(B_\mathcal{B})
    &=4(\kappa-3\nu/4)(\kappa-3\nu/4-1)(\kappa-\nu/2-2)\\
    &+4(\kappa-3\nu/4)(\nu/4)(\kappa-\nu/2-1)\\
    &+4(\nu/4)(\kappa-3\nu/4)(\kappa-\nu/2-1)\\
    &+4(\nu/4)(\nu/4)(\kappa-\nu/2).
\end{split}
\end{align}
In the unbalanced construction, we have $t_{\{0,1,2\}}=t_{\{0,1,3\}}=t_{\{0,2,3\}}=\kappa-\nu$, $t_{\{1,2,3\}}=\kappa$, $t_{\{0,1\}}=t_{\{0,2\}}=t_{\{0,3\}}=\kappa-\nu$, and  $t_{\{1,2\}}=t_{\{1,3\}}=t_{\{2,3\}}=\kappa$, thus
\begin{align}
\label{Eq:F_Unbalgamma4}
\begin{split}
    F(B_\mathcal{U})
    &\!=\!3(\kappa\!-\!\nu)(\kappa\!-\!\nu\!-\!1)(\kappa\!-\!2)\!+\!\kappa(\kappa\!-\!1)(\kappa\!-\!2).
\end{split}
\end{align}
In view of (\ref{Eq:F_Balgamma4}) and (\ref{Eq:F_Unbalgamma4}),  $F(B_\mathcal{B})-F(B_\mathcal{U})=\nu^2(3/2-\nu)<0
$ since $\nu \geq \gamma_L =4$.
\end{proof}
\begin{remark}
In Proposition~\ref{pro_FB_FU_gamma4}, we assumed $\nu$ is divisible by $\gamma_L$ only for simplicity. One can find a condition on $\nu$ for general case $\nu=a\gamma_L+b$ (where $a\geq 0$ and $0\leq b<\gamma_L$) such that $F(B_\mathcal{U})>F(B_\mathcal{B})$, by formulating the overlap parameters in terms of parameters $a$, $b$, and $\kappa$.
\end{remark}
\begin{remark}
For $\gamma_L=3$, there is a trade-off between cycle and threshold properties of local codes, and it is the designer discretion to choose between balanced and unbalanced schemes, depending on which feature is more desirable. This trade-off does not exist for $\gamma_L=4$, where the balanced scheme has better performance in both features.
\end{remark}

\section{Global Design}
\label{Sec:CC_opt}

In this section, we address the following question: given $\gamma_L$ LCNs, how one should design CCNs, i.e., entries in first $\gamma_C$ rows of the partitioning matrix $P$, in order to reduce the population of short cycles in the global code? The case of SC codes with no locality was optimally solved in \cite{EsfHar19}; as we will see, adding locality requires new considerations that convert the original problem of the optimal overlap partitioning to a well-defined constrained optimal overlap partitioning. 
We mark by $P_C$ and $P_L$ the upper $\gamma_C$ and lower $\gamma_L$ rows of $P$ (see Section~\ref{Sub:SC-LDPCL}), respectively, and assume that $P_L$ is given. We study three partitioning methods for determining $P_C$: 
\begin{enumerate}
\item\label{Item:CV_method} \textit{Cutting-vector (CV) partitioning} \cite{MitchellISIT2014}: let $0<\zeta_0<\zeta_1<\ldots<\zeta_{\gamma_C-1}$ be natural numbers. Set $[P_C]_{i,j}=1$ if and only if $j<\zeta_i$. In this paper, CV partitioning is used as a reference, and we consider uniform cutting vectors where $\zeta_{k}-\zeta_{k-1}$ is the the same for every $k\in\{0,\dots,\gamma_C-2\}$ (up to a residue due to possible indivisibility of $\kappa$ by $\gamma_C$). 
\item \label{Item:Unconstrained_OO} 
\textit{Locality-blind optimal (LBO) partitioning}: the optimal overlap partitioning for an SC code with $\gamma=\gamma_C$ (see Section~\ref{Sub:cycles_opt}). In other words, we are blind to the presence of LCNs that are already assigned to $B_0$, and optimize $P_C$ as there is no $P_L$.
\item \label{Item:Constrained_OO} \textit{Locality-aware optimal (LAO) partitioning}: the optimal overlap partitioning for an SC code with $\gamma=\gamma_C+\gamma_L$ and $P_L$ given as a constraint.
\end{enumerate}

In what follows, we focus on regular codes, i.e., $P_L$ is an all-zero matrix. Identifying the optimal partitioning is notably simpler with this assumption compared to cases with possibility of zero circulants, as Lemma~\ref{equ_O_ind_SCL} confirms. After optimization, we can replace the local code with an irregular code suggested in Section~\ref{Sec:local_des}. Recall that the rate of an SC code depends on the rate of the underlying code and the coupling termination, and does not depend on the partitioning.

\begin{lemma}\label{equ_O_ind_SCL} 
The set of independent non-zero overlap parameters for SC codes with $P_L=\mathbf{0}(\gamma_L,\kappa)$ is:
\begin{align*}
\mathcal{O}_\textnormal{ind}=\{t_{\{i_1,\dots,i_d\}}\textnormal{ }|\textnormal{ }1\leq d\leq \gamma_C,0\leq i_1,\dots,i_d<\gamma_C\}.
\end{align*}
The overlap parameters that are not included in $\mathcal{O_\textnormal{ind}}$ are either zero or functions of the overlap parameters in $\mathcal{O_\textnormal{ind}}$.
\end{lemma}

\begin{proof} First we assume $0\leq i_1,\cdots,i_{d_1}\leq \gamma-1$, $\gamma \leq j_1,\cdots,j_{d_2}\leq 2\gamma-1$, and $1\leq(d_1+d_2)\leq \gamma$. Then, as shown in \cite[Lemma~3]{EsfHar19} with $m=1$, $t_{\{i_1,\cdots,i_{d_1},j_1,\cdots,j_{d_2}\}}$ is a linear function of the overlap parameters in \vspace{-0.15cm}
\begin{align*}
\mathcal{O}_\textnormal{ind}'=\{t_{\{i_1,\dots,i_d\}}\textnormal{ }|\textnormal{ }1\leq d\leq \gamma,0\leq i_1,\dots,i_d<\gamma\}\,.
\end{align*}
Next, we assume $0\leq i_1,\cdots,i_{d_1}\leq \gamma_C-1$, $\gamma_C \leq j_1,\cdots,j_{d_2}\leq \gamma-1$, and $1\leq(d_1+d_2)\leq \gamma$. Then, 
$t_{\{i_1,\dots,i_{d_1},j_1,\dots,j_{d_2}\}}=t_{\{i_1,\dots,i_{d_1}\}}\,.$ This follows since all elements in rows $\{\gamma_C,\dots,\gamma-1\}$ of $B_0$ are $1$s, and thus the value of a degree-$d$ overlap parameter that is defined over a set of rows that includes some rows $j\in\{\gamma_C,\dots,\gamma-1\}$ is equal to the value of the overlap parameter when those rows are excluded.\vspace{-0.2cm}
\end{proof}

According to Lemma~\ref{equ_O_ind_SCL}, the number of independent overlap parameters is a function of $\gamma_C$ not $\gamma=\gamma_C+\gamma_L$. Thus, the complexity of LAO partitioning with $P_L=\mathbf{0}(\gamma_L,\kappa)$, does not increase when the SC-LDPC code features sub-block locality with regular local codes. 

\section{Simulation Results}
\label{Sec:Simulations}
In our simulations, we consider parameters $\kappa=p=13$, $\gamma_C=\gamma_L=3$, $m=1$, $L=10$, and AB lifting that yields cycle-4-free graphs. We investigate the performance of local and global decoding of SC-LDPC codes with sub-block locality constructed using various methods (new methods introduced in this paper and existing methods). Our results include the BER performances, cycle counts, and threshold values.

Let SC~Code~1, SC~Code~2, and SC~Code~3 be SC-LDPC codes with sub-block locality with the parameters given above, $P_L=\mathbf{0}(\gamma_L,\kappa)$, and constructed using CV, LBO, and LAO, respectively, as follows:
\begin{align*}
    P_{C,CV}\hspace{-0.05cm}&=\hspace{-0.15cm}\left(\begin{array}{ccccccccccccc}
        0&0&0&1&1&1&1&1&1&1&1&1&1 \\
        0&0&0&0&0&0&1&1&1&1&1&1&1 \\
        0&0&0&0&0&0&0&0&0&1&1&1&1 
    \end{array}\hspace{-0.1cm} \right)\hspace{-0.1cm},\\
    P_{C,LBO}\hspace{-0.05cm}&=\hspace{-0.15cm}\left(\begin{array}{ccccccccccccc}
        0&1&1&1&1&1&1&1&1&1&1&1&1 \\
        1&0&0&0&0&0&0&1&1&1&1&1&1 \\
        0&1&1&1&1&1&1&0&0&0&0&0&0 
    \end{array}\hspace{-0.1cm} \right)\hspace{-0.1cm},\\
     P_{C,LAO}\hspace{-0.05cm}&=\hspace{-0.15cm}\left(\begin{array}{ccccccccccccc}
        0&1&0&1&0&1&0&1&0&1&0&1&1 \\
        1&0&1&0&1&0&1&0&1&0&1&0&0 \\
        0&0&1&0&1&0&1&0&1&0&1&1&1 
    \end{array}\hspace{-0.1cm} \right)\hspace{-0.1cm}.
\end{align*}

Next, we add irregularity to the local codes and define SC~Code~4 and SC~Code~5 with the given parameters, $P_C=P_{C,LAO}$, and $\nu=10$. Consider the protograph matrix $B$ of a local code with dimensions $\gamma_L=3$ and $\kappa=13$. The matrix $P_L$ has the same dimensions as $B$, ${\left(P_L\right)}_{i,j}=X$ when $B_{i,j}=0$, and ${\left(P_L\right)}_{i,j}=0$ when $B_{i,j}=1$. SC~Code~4 has $P_L$ constructed from balanced matrix $B_\mathcal{B}$ defined in (\ref{def_mat_BB}) and SC~Code~5 has $P_L$ constructed from unbalanced matrix $B_\mathcal{U}$ defined in (\ref{def_mat_BU}). Let LC~Code~1 and LC~Code~2 represent the local codes for SC~Code~4 and SC~Code~5, respectively.

The population of cycles-$6$ and cycles-$8$ in the protographs and lifted graphs along with the threshold values are given in Table~\ref{Tab:gamL_GamC_3_gamL3_cycle_pop_threshold}. According to the results, the LAO method yields about $21\%$ reduction in the population of cycles-$6$ (both in protographs and lifted graphs) compared to the CV method, while this reduction is less than $5\%$ for the LBO method compared to the CV method. By removing $\nu=10$ circulants from local codes, we achieve further reductions in the number cycles-$6$, i.e., $63\%$ and $55\%$ for the LAO method with balanced and unbalanced irregularities, respectively, compared to the LAO method with $\nu=0$. In terms of asymptotic behavior, the local threshold of the balanced code (LC code 1) is higher than the local threshold of unbalanced code (LC code 2) as Proposition~\ref{Prop:thU<thB} predicts. In addition, the global thresholds of the SC codes when using irregular local codes (SC~Codes~4-5) are higher than the regular SC code (SC~Codes~1-3).

Note that the balanced method for adding irregularities results in both better global threshold and lower cycles-$6$ population for SC-LDPC codes with sub-block locality. However, for local decoding, there is a trade-off and the unbalanced scheme results in lower population of cycles-$6$ but also worse threshold compared to the balanced scheme.

\begin{table}
\centering
\caption{Cycle population and threshold ($C_k$ is cycle-$k$)\vspace{-0.2mm}}
\setlength\tabcolsep{4pt}
\begin{tabular}{c|c|c|c|c|c}
& proto $C_6$ & lifted  $C_6$ & proto $C_8$ & lifted $C_8$  & $\sigma^*$\\
\hline
\hline
SC~Code~1 & $173{,}232$ & $204{,}698$ & $3{,}741{,}840$ & $7{,}410{,}481$ & $0.8283$ \\
SC~Code~2 & $165{,}120$ & $195{,}624$ & $3{,}309{,}696$ & $7{,}161{,}258$ & $0.7995$ \\
SC~Code~3 & $137{,}362$ & $162{,}084$ & $2{,}957{,}941$ & $5{,}957{,}055$ & $0.8059$ \\
\hline
SC~Code~4 & $48{,}647$ & $59{,}202$ & $861{,}740$ & $1{,}560{,}143$ & $0.8382$\\
SC~Code~5 & $60{,}812$ & $72{,}267$ & $1{,}041{,}381$ & $2{,}284{,}048$ & $0.8373$ \\
\hline
LC~Code~1 & $201$ & $273$ & $0$ & $3{,}313$ & $0.5542$ \\
LC~Code~2 & $66$ & $78$ & $0$ & $9{,}014$ & $0.4961$
\end{tabular} 
\label{Tab:gamL_GamC_3_gamL3_cycle_pop_threshold}
\end{table}

Fig.~\ref{Fig:GlobalDecoding} compares the global-decoding performance for SC~Codes~1--5 over the AWGN channel. The plot shows the superiority of LAO partitioning for all SNR values, e.g., more than $1.5$ orders of magnitude compared to the CV method at SNR$=7$~dB. In addition, it shows that the LBO partitioning is inferior even to CV partitioning. Thus, when one adds locality considerations, one must re-design the global code as well. Moreover, adding irregularity improves the performance, and the balanced design outperforms the unbalanced design, e.g., more than $1.2$ orders of magnitude at SNR$=6.5$~dB.

\vspace{-1mm}

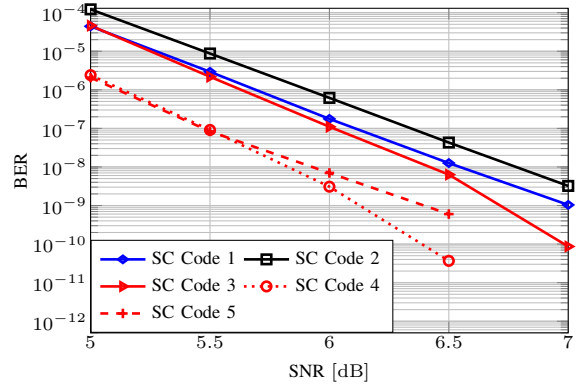
\begin{figure}[h!]
\centering
\begin{tikzpicture}

\begin{axis}[%
width=2.5in,
height=1.7in,
at={(0,0)},
scale only axis,
xmin=5,
xmax=7,
xlabel style={font=\scriptsize,yshift=2.5mm},
xlabel={SNR $[\mathrm{dB}]$},
xticklabel style = {font=\scriptsize,yshift=0.5ex},
xtick={5,5.5,...,7},
ymode=log,
ymin=5e-13,
yticklabel style = {font=\scriptsize,xshift=0.5ex},
ymax=0.00013,
yminorticks=true,
max space between ticks=10,
ylabel style={font=\scriptsize,yshift=-3mm},
ylabel={BER},
axis background/.style={fill=white},
xmajorgrids,
ymajorgrids,
yminorgrids,
legend columns=2, 
legend style={
    legend cell align = left,
    align = left, 
    anchor = south west,
    at = {(0,0)},
    font = \scriptsize,
    /tikz/column 2/.style={column sep=0.5pt},
    }
]
\addplot [color=blue, line width=1pt, mark size=1.8
pt, mark=diamond, mark options={solid, rotate=270, blue}]
  table[row sep=crcr]{%
5	4.44053321601707e-05\\
5.5	2.93350287269817e-06\\
6	1.76437869822485e-07\\
6.5	1.26065088757396e-08\\
7	1.03786982248521e-09\\
};
\addlegendentry{SC Code 1}

\addplot [color=black, line width=1pt, mark size=1.8pt, mark=square, mark options={solid, rotate=270, black}]
  table[row sep=crcr]{%
5	0.000121954277678194\\
5.5	8.74111676859768e-06\\
6	6.18082840236686e-07\\
6.5	4.29428007889546e-08\\
7	3.22390532544379e-09\\
};
\addlegendentry{SC Code 2}

\addplot [color=red, line width=1pt, mark size=1.8pt, mark=triangle, mark options={solid, rotate=270, red}]
  table[row sep=crcr]{%
5	4.62180995706157e-05\\
5.5	2.18311669131337e-06\\
6	1.10298224852071e-07\\
6.5	6.33372781065089e-09\\
7	8.67286559594252e-11\\
};
\addlegendentry{SC Code 3}

\addplot [color=red, dotted, line width=1pt, mark size=1.8pt, mark=o, mark options={solid, red}]
  table[row sep=crcr]{%
5	2.40089691534617e-06\\
5.5	9.10567372105834e-08\\
6	3.08090179551518e-09\\
6.5	3.67788025582828e-11\\
};
\addlegendentry{SC Code 4}

\addplot [color=red, dashed, line width=1pt, mark size=1.8pt, mark=+, mark options={solid, red}]
  table[row sep=crcr]{%
5	2.07715737117249e-06\\
5.5	8.34934911242604e-08\\
6	6.99526627218935e-09\\
6.5	6e-10\\
};
\addlegendentry{SC Code 5}

\end{axis}
\end{tikzpicture}%
\vspace{-0.2cm}
\caption{\label{Fig:GlobalDecoding}Global-decoding BER curves for the proposed SC codes with $\gamma_C=3$, $\gamma_L=3$, $\kappa=p=13$, $m=1$, $L=10$, over the AWGN channel.\vspace{-0.2cm}}
\end{figure}

Fig.~\ref{Fig:LocalDecoding} compares the local-decoding performance of LC~Codes~1-2 over the AWGN channel. In the low-SNR regime the balanced construction is superior over the unbalanced one, while in the high-SNR regime the trend is opposite.
This observation is consistent with  Propositions~\ref{Prop:thU<thB} and \ref{Prop:C6U<C6B} \footnote{The difference will get more prominent if we increase the SNR. Due to the complexity of collecting BER points in the deep error floor region, we were not able to exemplify this.}\vspace{-0.2cm}.

\vspace{-1mm}

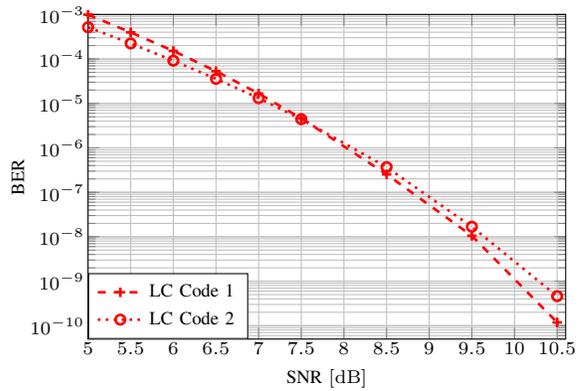
\begin{figure}[h!]
\centering
\begin{tikzpicture}

\begin{axis}[%
width=2.5in,
height=1.7in,
at={(0,0)},
scale only axis,
xmin=5,
xmax=10.6,
xlabel style={font=\scriptsize,yshift=2.5mm},
xlabel={SNR $[\mathrm{dB}]$},
xticklabel style = {font=\scriptsize,yshift=0.5ex},
xtick={5,5.5,...,10.5},
ymode=log,
ymin=5e-11,
ymax=1e-3,
yticklabel style = {font=\scriptsize,xshift=0.5ex},
max space between ticks=10,
yminorticks=true,
ylabel style={font=\scriptsize,yshift=-3mm},
ylabel={BER},
axis background/.style={fill=white},
xmajorgrids,
ymajorgrids,
yminorgrids,
legend style={
    legend cell align = left,
    align = left, 
    anchor = south west,
    at = {(0,0)},
    font = \scriptsize,
    row sep=0.01pt
    }
]
\addplot [color=red, dashed, line width=1pt, mark size=1.8pt, mark=+, mark options={solid, red}]
  table[row sep=crcr]{%
5	0.000989980242284304\\
5.5	0.000396922783297951\\
6	0.000151595709693521\\
6.5	5.3174332843336e-05\\
7	1.66752353805254e-05\\
7.5	4.61995781117573e-06\\
8.5	2.53798816568047e-07\\
9.5	1.0508875739645e-08\\
10.5	1.18343195266272e-10\\
};
\addlegendentry{LC Code 1}

\addplot [color=red, dotted, line width=1pt, mark size=1.8pt, mark=o, mark options={solid, red}]
  table[row sep=crcr]{%
5	0.000517269967926235\\
5.5	0.000223371501618118\\
6	9.15806964008017e-05\\
6.5	3.57331706570701e-05\\
7	1.31436135110951e-05\\
7.5	4.42086761785788e-06\\
8.5	3.6870241952038e-07\\
9.5	1.68165680473373e-08\\
10.5	4.61538461538462e-10\\
};
\addlegendentry{LC Code 2}

\end{axis}
\end{tikzpicture}%
\vspace{-0.2cm}
\caption{\label{Fig:LocalDecoding}Local-decoding BER curves for balanced (LC Code 1) and unbalanced (LC Code 2) codes with $\gamma_L=3$, $\kappa=p=13$, over the AWGN channel.\vspace{-0.2cm}}
\end{figure}

For Monte Carlo simulations depicted in Fig.~\ref{Fig:GlobalDecoding} and Fig.~\ref{Fig:LocalDecoding}, we observed at least $50$ frame errors in all collected points.

\section{Acknowledgment} Research supported in part by a grant from ASRC-IDEMA, NSF CCF-BSF:CIF 1718389, and ISF 2525/19.

\bibliographystyle{IEEEtran}
\bibliography{IEEEabrv,references}

\end{document}